\documentclass[letterpaper, 10 pt, conference]{ieeeconf}  

\IEEEoverridecommandlockouts                              
\usepackage{graphics} 
\usepackage{mathptmx} 
\usepackage{times} 
\usepackage{amsmath} 
\usepackage{amssymb}  

\usepackage{mathptmx} 
\usepackage{times} 
\usepackage{amsmath} 
\usepackage{amssymb}  

\usepackage{mathtools}

\mathtoolsset{centercolon}

\usepackage{siunitx}
\usepackage{balance}

\usepackage[noadjust,sort]{cite}

\newtheorem{definition}{Definition}[section]
\newtheorem{assumption}[definition]{Assumption}
\newtheorem{lemma}[definition]{Lemma}
\newtheorem{proposition}[definition]{Proposition}
\newtheorem{theorem}[definition]{Theorem}

\newtheorem{example}[definition]{Example}
\newtheorem{problem}[definition]{Problem}

\usepackage{todonotes}

\usepackage{accents}
\newcommand{\ubar}[1]{\underaccent{\bar}{#1}}

\usepackage{calrsfs}
\DeclareMathAlphabet{\pazocal}{OMS}{zplm}{m}{n}
\renewcommand{\mathcal}[1]{\pazocal{#1}}
\newcommand{\tdsp}{\hspace{.75pt}}

\usepackage{color}

\newcommand{\afterequation}{\vskip 3pt}

\DeclareMathOperator*{\minimize}{minimize}
\DeclareMathOperator*{\find}{find}
\DeclareMathOperator*{\subjectto}{subject\ to}

\newenvironment{proofof}[1]{
\renewcommand\proof{\noindent\hspace{2em}{\itshape{Proof of #1: }}}%
\begin{proof}}{\end{proof}
}

\title{\LARGE \bf
Resource Allocation for Containing Epidemics\\from Temporal Network Data
}

\author{Masaki Ogura and Junichi Harada
\thanks{The authors are with the Graduate School of Information Science, Nara Institute of Science and Technology, Ikoma, Nara, Japan. 
        email: {\tt\small \{oguram, harada.junichi.hh3\}@is.naist.jp}}%
}

\begin{document}

\maketitle
\thispagestyle{empty}
\pagestyle{empty}

\begin{abstract}
We study the problem of containing epidemic spreading processes in temporal networks. We specifically focus on the problem of finding a resource allocation to suppress epidemic infection, provided that an empirical time-series data of connectivities between nodes is available. Although this problem is of practical relevance, it has not been clear how an empirical time-series data can inform our strategy of resource allocations, due to the computational complexity of the problem. In this direction, we present a computationally efficient framework for finding a resource allocation that satisfies a given budget constraint and achieves a given control performance. The framework is based on convex programming and, moreover, allows the performance measure to be described by a wide class of functionals called posynomials with nonnegative exponents. We illustrate our theoretical results using a data of temporal interaction networks within a primary school.
\end{abstract}

\section{INTRODUCTION}

The containment of epidemic spreading processes taking place on complex networks is a major research area in the network science~\cite{Newman2006}. Relevant applications include information spread in on-line social networks, the evolution of epidemic outbreaks in human contact networks, and the dynamics of cascading failures in the electrical grid. Important advances in the analysis and containment of spreading processes over \emph{static} networks have been made during the last decade~\cite{Nowzari2015a,Pastor-Satorras2015a}. For example, Cohen et al.~\cite{Cohen2003} proposed a heuristic vaccination strategy called an acquaintance immunization policy and showed proved it to drastically improve the random vaccine distribution. The problem of determining the optimal allocation of control resources over static networks to efficiently eradicate epidemic outbreaks has been studied in~\cite{Preciado2014}. An efficient curing policy based on graph cuts has been proposed in~\cite{Drakopoulos2014}. Decentralized algorithms for epidemic control have been proposed in~\cite{Trajanovski2015a}. Other approaches based on the control theory can be found in, e.g.,~\cite{Wan2008IET,Khouzani2012a}. Recently,  cost-efficiency of various heuristic vaccination strategies were thoroughly investigated in~\cite{Holme2017}.

On the other hand, most epidemic processes of practical interest take place on \emph{temporal networks}~\cite{Masuda2013} having time-varying topologies~\cite{Holme2015b}. Although major advances have been made for the \emph{analysis} of epidemic spreading processes over temporal networks (see, e.g.,~\cite[Section~VIII]{Pastor-Satorras2015a} and~\cite[Section 6.4]{Masuda2016b}), there is still scarce of methodologies for \emph{containing} epidemic outbreaks on temporal networks. In this direction, Lee et al.~\cite{Lee2012} have presented heuristic vaccination strategies that exploit temporal correlations. Liu et al.~\cite{Liu2014a} have proposed an immunization strategy for a class of temporal networks called the activity-driven networks~\cite{Perra2012}. Optimization frameworks for distributing containment resources have been proposed for Markovian~\cite{Ogura2015a,Nowzari2015b} and adaptive~\cite{Ogura2015i} temporal networks. However, it is still left as an open problem how to effectively fit an empirical dataset of temporal networks to these stochastic models of temporal networks. Furthermore, the aforementioned results focus on the asymptotic evolution of epidemic infections and, therefore, do not allow us to control the evolution of epidemic spreading in a finite time window.

In this paper, we present an optimization framework for allocating control resources for eradicating epidemic infections in empirical temporal networks. We specifically show that, given a time-series data of temporal network, a budget constraint, and a requirement on control performance, we can find a resource allocation satisfying both the constraint and the requirement by solving a convex feasibility problem. Unlike in the aforementioned results, we allow the performance measure to depend on the transient evolution of epidemic processes. In order to realize this flexibility, we extend the class of functions called posynomials (see, e.g.,~\cite{Boyd2007}) to function spaces. We numerically illustrate the obtained theoretical results using the empirical temporal network between the children and teachers in a primary school~\cite{Stehle2011}.

This paper is organized as follows. In Section~\ref{sec:problemSetting}, we introduce the model of epidemic infection over temporal networks and state the resource allocation problem studied in this paper. In Section~\ref{sec:mainResult}, we state our main result that reduces the resource allocation problem to a convex feasibility problem. The proof of the reduction is presented in Section~\ref{sec:proof}. We finally illustrate the effectiveness of our results via numerical simulations in Section~\ref{sec:simulations}.

\subsection{Mathematical Preliminaries}

Let $\mathbb R$ and $\mathbb{R}_{+}$ denote the set of real and positive numbers, respectively. A real matrix $A$, or a vector as its special case, is said to be nonnegative, denoted by~$A\geq 0$, if $A$ is nonnegative entry-wise. For another matrix $B$, we write $A\leq B$ if $B-A \geq 0$.

An undirected graph is a pair~$G = (V, E)$, where $V = \{v_1, \ldots, v_n\}$ is the set of nodes, and $E$ is the set of edges consisting of distinct and unordered pairs~$\{v_i, v_j\}$ for $v_i, v_j\in V$. We say that $v_j$ is a neighbor of~$v_i$ (or, $v_i$ and $v_j$ are adjacent) if $\{v_i, v_j\} \in E$. The adjacency matrix~$A\in \mathbb{R}^{n\times n}$ of~$G$ is defined as the $\{0, 1\}$\nobreakdash-matrix whose $(i,j)$ entry is equal to one if and only if $v_i$ and~$v_j$ are adjacent.

For a subset $X$ of $\mathbb{R}^{n\times m}$ and $T>0$, we let $L^\infty([0, T], X)$ denote the space of $X$-valued, Lebesgue-measurable, and essentially bounded functions on $[0, T]$.

\section{PROBLEM SETTING}\label{sec:problemSetting}

In this section, we introduce our model of disease spread over temporal networks. We then formulate the resource distribution problem studied in the paper. The computational difficulty of the problem is also discussed.

\subsection{SIS Model over Temporal Networks}

We start by reviewing a model of spreading processes over static networks called the susceptible--infected--susceptible (SIS) model~\cite{Pastor-Satorras2015a}. Let $G = (V, E)$ be an undirected graph, where nodes in~$V = \{v_1, \dots,  v_n\}$ represent individuals and edges in $E$ represent interactions between them. At a given time~$t \geq 0$, each node can be in one of two possible states: \emph{susceptible} or \emph{infected}. In the SIS model, when a node~$v_i$ is infected, it can randomly transition to the susceptible state with an instantaneous rate~$\delta_i > 0$, called the \emph{recovery rate} of node~$v_i$. On the other hand, if a neighbor of node~$v_i$ is in the infected state, then the neighbor can infect node~$v_i$ with the instantaneous rate~$\beta_i$, where $\beta_i > 0$ is called the \emph{transmission rate} of node $v_i$.
Therefore, if we define the variable
\begin{equation}\label{eq:def:x_i}
x_i(t) = \begin{cases}
0, & \mbox{if $v_i$ is susceptible at time $t$, }
\\
1, & \mbox{if $v_i$ is infected at time $t$, }
\end{cases}
\end{equation}
then the transition probabilities of the SIS model in the time window $[t, t+h]$ can be written as
\begin{align}
\Pr(x_i(t+h) = 0 \mid x_i(t) = 1) 
&= 
\delta_i h + o(h), \label{eq:recovery}
\\
\Pr(x_i(t+h) = 1 \mid x_i(t) = 0) 
&=
\beta_i \sum_{{j \in N_i}} x_j(t) h + o(h),\notag 
\end{align}
where $N_i$ denotes the set of neighbors of~$v_i$ and $o(h)/h\to 0$  as $h \to 0$.

The SIS model over static networks can be naturally extended to the case of temporal networks (i.e., time-varying networks)~\cite{Ogura2015c,Masuda2016b}. In this paper, we adopt the following definition of temporal networks:

\begin{definition}\label{defn:temporal}
Let $T>0$ and a set of nodes~$V$ be given. A piecewise-constant and right-continuous function defined on $[0, T]$ and taking values in the set of undirected networks having nodes $V$ is called a \emph{temporal network}.
\end{definition}

As in the case of static networks, at each time $t\geq 0$, each node can be either susceptible or infected in the SIS model over temporal networks.  For all $i \in \{1, \ldots, n\}$ and $t \in [0, T]$, let us define the variable~$x_i(t)$ by~\eqref{eq:def:x_i}. Then, we define the transition probabilities of the SIS model over the temporal network by~\eqref{eq:recovery} and
\begin{equation*}
\Pr(x_i(t+h) = 1 \mid x_i(t) = 0) 
=
\beta_i \sum_{{j \in N_i(t)}} x_j(t) h + o(h), 
\end{equation*}
where $N_i(t)$ denotes the set of neighbors of node~$v_i$ at time~$t$.

\subsection{Problem Formulation}

Let us consider the following epidemiological problem~\cite{Preciado2014}: Suppose that we can use vaccines for reducing the transmission rates of individuals in the network, and antidotes for increasing their recovery rates. Assuming that vaccines and antidotes have an associated cost and that we are given a fixed budget, how should we distribute vaccines and antidotes throughout the individuals in the temporal network to suppress epidemic infections? 

In order to rigorously state this problem, define the infection probability 
\begin{equation*}
p_i(t) = P(\mbox{$v_i$ is infected at time $t$})
\end{equation*}
and the vector 
\begin{equation*}
p = 
{\begin{bmatrix}
p_1 & \cdots & p_n
\end{bmatrix}}^\top. 
\end{equation*}
Suppose that we are given a functional $J \colon L^\infty([0, T], \mathbb{R}_{+}^n) \to \mathbb{R}_{+}$ to measure the persistence of epidemic infection. For achieving a small value of~$J(p)$, we assume~\cite{Preciado2014} that the transmission and recovery rates can be tuned within the following intervals:
\begin{equation}\label{eq:boxConstraints}
0< \ubar \beta_i \leq \beta_i \leq \bar \beta_i, \ 
0< \ubar \delta_i \leq
\delta_i \leq \bar \delta_i. 
\end{equation}
Furthermore, suppose that we have to pay $\phi_i(\beta_i)$ unit of cost to tune the transmission rate of node~$v_i$ to $\beta_i$. Likewise, we assume that the cost for tuning the recovery rate of node~$v_i$ to~$\delta_i$ equals $\psi_i(\delta_i)$. Notice that the total cost of realizing the collection of transmission rates~$\beta = (\beta_1, \ldots, \beta_n)$ and recovery rates $\delta = (\delta_1, \ldots, \delta_n)$ in the network  is given by
\begin{equation*}
R(\beta, \delta) = \sum_{i=1}^n \left(\phi_i(\beta_i) + \psi_i(\delta_i)\right).
\end{equation*}

We can now state our resource allocation problem. 

\begin{problem}\label{prb:original}
Given a temporal network~$\mathcal G$, an initial condition~$p(0)\in \mathbb{R}_{+}^n$, and positive constants~$\bar J$ and~$\bar R$, find the transmission and recovery rates~$\beta$ and $\delta$ satisfying the feasibility constraints~\eqref{eq:boxConstraints}, the performance constraint
\begin{equation}\label{eq:performanceConst}
J(p) \leq \bar J, 
\end{equation}
and the budget constraint
\begin{equation}\label{eq:budgetConstraint}
R(\beta, \delta) \leq \bar R. 
\end{equation}
\afterequation
\end{problem}

As is well known~\cite[Section~IV]{Pastor-Satorras2015a}, it is not practically feasible to even evaluate the infection probabilities $p$ for large-scale networks. To briefly illustrate the difficulty, let us focus on the SIS model over a static network. Observe that the collection of variables $(x_1, \ldots, x_n)$ is a Markov process having the total of~$2^n$ possible states (two states per node). Let us label the $2^n$ states as $s_1$, \dots, $s_{2^n}$, and let $q_\ell(t)$ denote the probability that the process is in the state~$s_\ell$ at time~$t$. Then, the infection probability~$p_i(t)$ is equal to a linear combination of the probabilities~$q_1(t)$, \dots,~$q_{2^n}(t)$. However, the computation of all the probabilities $q_1(t)$, \dots,~$q_{2^n}(t)$ is demanding for large-scale networks. Since the computational difficulty is inherited in the case of temporal networks, it is not realistic to directly solve Problem~\ref{prb:original} for large-scale temporal networks.

\section{MAIN RESULTS}\label{sec:mainResult}

In this section, we present a solution to Problem~\ref{prb:original} in terms of a convex feasibility problem. In Subsection~\ref{subsec:posy}, we introduce a novel class of functionals called posynomials with nonnegative exponents and extend them to functionals on function spaces. Under the assumption that the objective function belongs to this class, in Subsection~\ref{subsec:cert} we show that the solution of Problem~\ref{prb:original} can be given by solving a convex feasibility problem. We finally discuss some optimal resource allocation problems in Subsection~\ref{subsec:optimal}.

\subsection{Posynomials with Nonnegative Exponents}\label{subsec:posy}

We start by reviewing the notion of posynomials and monomials~\cite{Boyd2007}. Let $F \colon \mathbb{R}^n_{+}\to \mathbb{R}_{+}$ be a function. We say that $F$ is a \emph{monomial} if there exist $c>0$ and real numbers~$a_1$, \dots,~$a_n$ such that
\begin{equation}\label{eq:f:monomial}
F(v) = c v_{\mathstrut 1}^{a_{1}} \dotsm v_{\mathstrut n}^{a_n}. 
\end{equation}
We say that $F$ is a \emph{posynomial} if $F$ is a sum of finitely many monomials. We say that $F$ is a \emph{generalized posynomial} if $F$ can be formed from posynomials using the operations of addition, multiplication, positive (fractional) power, and maximum. The following lemma shows the log-log convexity of posynomials and is used for the proof of our main results:

\begin{lemma}[\cite{Boyd2007}]\label{lem:llconvexity}
Let $F \colon \mathbb{R}_{+}^n \to \mathbb{R}_{+}$ be a generalized posynomial. Define the function $f\colon \mathbb{R}^n \to \mathbb{R}$ by
$f(w) = \log F(\exp[w])$, 
where $\exp[\cdot]$ denotes the entry-wise exponentiation of vectors. Then, $f$ is convex. 
\end{lemma}

In this paper, the following class of monomials and posynomials plays an important role. 

\begin{definition}
Let $F \colon \mathbb{R}^n_{+}\to \mathbb{R}_{+}$ be a function. 
\begin{itemize}
\item We say that $F$ is a \emph{monomial with nonnegative exponents} if there exist $c>0$ and nonnegative numbers~$a_1$, \dots,~$a_n$ such that~\eqref{eq:f:monomial} holds true. 

\item We say that $F$ is a \emph{posynomial with nonnegative exponents} if $F$ is the sum of finitely many monomials with nonnegative exponents. 

\item We say that $F$ is a \emph{generalized posynomial with nonnegative exponents} if $F$ can be formed from posynomials with nonnegative exponents using the operations of addition, multiplication, positive (fractional) power, and maximum. 
\end{itemize}
\end{definition}

We further extend this definition to functionals on function spaces: 

\begin{definition}
Let $F \colon L^\infty([0, T], \mathbb{R}^n_{+}) \to \mathbb{R}_{+}$ be a functional. 
\begin{itemize}
\item We say that $F$ is a \emph{finite-monomial with nonnegative exponents} if there exist $t_1, \ldots, t_m \in [0, T]$, $i_1, \ldots, i_m \in \{1, \ldots, n\}$, and a monomial $g\colon \mathbb{R}^m_{+} \to \mathbb{R}_{+}$ with nonnegative exponents such that $F(v)  = g(v_{i_1}(t_1), \dotsc, v_{i_m}(t_m))$.

\item We say that $F$ is a \emph{finite-posynomial with nonnegative exponents} if $F$ is a finite sum of finite-monomials with nonnegative exponents. 

\item We say that $F$ is a \emph{generalized finite-posynomial with nonnegative exponents} if $F$ can be formed from finite-posynomials with nonnegative exponents using the operations of addition, multiplication, positive (fractional) power, and maximum.

\item We say that $F$ is a \emph{generalized posynomial with nonnegative exponents} if $F$ is a pointwise limit of a sequence of generalized finite-posynomials with nonnegative exponents. 
\end{itemize}
\afterequation
\end{definition}

We now state our assumption on the performance measure~$J$. 

\begin{assumption}\label{asm:}
$J$ is a generalized posynomial with nonnegative exponents. 
\end{assumption}

This assumption allows us to describe several performance measures of interest, as illustrated below.

\begin{example}
Let $w_1$, \dots,~$w_n$ be positive numbers. Let $t \in [0, T]$ and $q > 0$ be arbitrary. Then, the weighted $\ell_q$-norm
\begin{equation*}
J(p) 
=
{\biggl(\sum_{i=1}^n {(w_i p_i(t))}^q\biggr)}^{1/q}
\end{equation*}
is a generalized finite-posynomial with nonnegative exponents and, therefore, satisfies Assumption~\ref{asm:}. The weights~$w_1, \ldots, w_n$ adjust the protection level of the nodes (i.e., the larger $w_i$, the stronger node $v_i$ will be protected). We can also tune the shape of the cost functional by changing the value of the exponent~$q$.
\end{example}

\begin{example}
Let $w\in L^\infty([0, T], \mathbb{R}^n_{+})$ and define 
\begin{equation*}
J(p) = \int_0^T w(t)^\top p(t) \,dt. 
\end{equation*}
Let us confirm that $J$ satisfies Assumption~\ref{asm:}. For each $k\geq 1$, let $h = T/k$ and define $J_k(p) =  \sum_{\ell=0}^{N} h w(\ell h)^\top p(\ell h)$. Then, $J_k$ is a finite-posynomial with nonnegative exponents for every $k$. Moreover, the measurability of~$w$ and~$p$ shows $J(p) = \lim_{k\to\infty} J_k(p)$. Therefore, $J$ is a generalized posynomial with nonnegative exponents. 
\end{example}

\subsection{Convex Feasibility Certificate}\label{subsec:cert}

This subsection presents the main result of this paper. We place on the cost functions the following assumptions~\cite{Preciado2014,Ogura2015i}. 

\begin{assumption}\label{asm:cost}
For all $i\in \{1, \ldots, n\}$, define the functions $\phi_i^+ = \max(\phi_i, 0)$,  $\phi_i^- = \max(-\phi_i, 0)$, $\psi_i^+ = \max(\psi_i, 0)$, and $\psi_i^- = \max(-\psi_i, 0)$. The following conditions hold true: 
\begin{itemize}
\item $\phi_i^+$ is a posynomial for all $i$; 
\item 
There exists $\hat \delta > \max(\bar \delta_1, \ldots, \bar \delta_n)$ such that the function
\begin{equation*}
\tilde \psi_i^+ \colon \mathbb{R}_{+} \to \mathbb{R}_{+} \colon 
\tilde \delta_i \mapsto \psi_i^+(\hat \delta - \tilde \delta_i)
\end{equation*}
is a posynomial for all $i$; 
\item $\phi_i^-$ and $\psi_i^-$ are nonnegative constants for all $i$. 
\end{itemize}
\afterequation
\end{assumption}

In order to state the main result, Let $\bar p$ denote the solution of the differential equation:
\begin{equation}\label{eq:barSigma}
\frac{d\bar p}{dt} = (BA(t) - D) \bar p, \quad \bar p(0) = p(0), 
\end{equation}
where $A(t)$ denotes the adjacency matrix of the network~$\mathcal G(t)$ for each $t \in [0, T]$, and the matrices $B$ and $D$ are the diagonal matrices having $\beta_1, \dotsc, \beta_n$ ($\delta_1, \dotsc, \delta_n$, respectively) as their diagonals. Let us denote by~$\bar p(\cdot; \beta, \delta) \in L^\infty([0, T], \mathbb{R}_{+}^n)$ the solution of the differential equation~\eqref{eq:barSigma} for transmission rates~$\beta =  (\beta_1, \ldots, \beta_n)$ and recovery rates~$\delta = (\delta_1, \ldots, \delta_n)$, and define
\begin{equation}\label{eq:def:F}
F(\beta, \delta) = J(\bar p(\cdot; \beta, \delta)). 
\end{equation}
Define $f\colon \mathbb{R}^{2n} \to \mathbb{R}$ by
\begin{equation}
f(b, \tilde d\tdsp ) = \log F(\exp[b], \hat \delta - \exp [\tilde d\tdsp ]).  \label{eq:def:f}
\end{equation}
Also, let 
\begin{equation*}
\begin{aligned}
R^+(\beta, \delta) = \sum_{i=1}^n \left(\phi_i^+(\beta_i) + \psi_i^+(\delta_i)\right), \     
R^- = \sum_{i=1}^n (\phi_i^- + \psi_i^-), 
\end{aligned}
\end{equation*}
and define $r^+\colon \mathbb{R}^{2n} \to \mathbb{R}$ by
\begin{equation*}
r^+(b, \tilde d\tdsp ) = \log R^+(\exp[b], \hat \delta - \exp [\tilde d\tdsp ]). 
\end{equation*}
The following theorem allows us to efficiently solve Problem~\ref{prb:original} and is the main result of this paper. We give the proof of the theorem in Section~\ref{sec:proof}.

\begin{theorem}\label{thm:main}
Solutions of Problem~\ref{prb:original} are given by 
\begin{equation}\label{eq:solutionExpressions}
\beta_i = \exp(b_i), \quad \delta_i = \hat \delta - \exp(\tilde d_i), 
\end{equation}
where $b \in \mathbb R^n$ and $\tilde d \in \mathbb R^n$ solve the following convex feasibility problem:
\begin{subequations}\label{eq:convexConstraints}
\begin{align}
\find\hspace{1.2cm}& b, \tilde d 
\\
\subjectto\quad&f(b, \tilde d\tdsp ) \leq \log \bar J, \label{eq:convexConstraints1}
\\
&r^+(b, \tilde d )\leq \log (\bar R + R^-), \label{eq:convexConstraints2}
\\
&\log \ubar{\beta}_i \leq b_i \leq \log \bar \beta_i, \label{eq:convexConstraints3}
\\
&\log (\hat \delta - \bar \delta_i) \leq \tilde d_i \leq \log (\hat \delta - \ubar \delta_i). \label{eq:convexConstraints4}
\end{align}
\end{subequations}
\afterequation
\end{theorem}

\subsection{Optimal Resource Allocation Problems}\label{subsec:optimal}

In this subsection, we formulate some optimal resource allocation problems and discuss how the problems can be sub-optimally solved using Theorem~\ref{thm:main}. We first consider the following performance-constrained allocation problem:

\begin{problem}\label{prb:pc}
Given a temporal network~$\mathcal G$, an initial condition $p(0)\in \mathbb{R}_{+}^n$, and a positive constant~$\bar J$, find the transmission and recovery rates~$\beta$ and $\delta$ satisfying the feasibility constraints~\eqref{eq:boxConstraints} and the performance constraint~\eqref{eq:performanceConst}, while minimizing the cost $R(\beta, \delta)$.
\end{problem}

Using Theorem~\ref{thm:main}, we can find sub-optimal solutions~\eqref{eq:solutionExpressions} to Problem~\ref{prb:pc} by solving the convex optimization problem: 
\begin{equation*}
\begin{aligned}
\minimize_{b,\,\tilde d}\ \  \,\,&  r^+(b, \tilde d\tdsp )
\\
\subjectto\ \  &\mbox{\eqref{eq:convexConstraints1}, \eqref{eq:convexConstraints3}, and~\eqref{eq:convexConstraints4}}. 
\end{aligned}
\end{equation*}

We also consider the budget-constrained allocation problem formulated as follows:

\begin{problem}\label{prb:bc}
Given a temporal network~$\mathcal G$, an initial condition~$p(0)\in \mathbb{R}_{+}^n$, and a positive constant~$\bar R$, find the transmission and recovery rates~$\beta$ and $\delta$ satisfying the feasibility constraints~\eqref{eq:boxConstraints} and the budget constraint~\eqref{eq:budgetConstraint}, while minimizing $J(p)$.
\end{problem}

In the same way as in the case of the budget-constrained allocation problem considered above, we can formulate the following convex optimization problem for finding sub-optimal solutions to Problem~\ref{prb:bc}:
\begin{equation}\label{eq:opt:bc}
\begin{aligned}
\minimize_{b,\,\tilde d}\ \ \,\, &  f(b, \tilde d\tdsp )
\\
\subjectto\ \  
&\mbox{\eqref{eq:convexConstraints2}, \eqref{eq:convexConstraints3}, and~\eqref{eq:convexConstraints4}}. 
\end{aligned}
\end{equation}
\afterequation

\section{PROOF} \label{sec:proof}

We give the proof of Theorem~\ref{thm:main} in this section. We start with the following lemma, which shows that the solution of the switched linear positive system~\eqref{eq:barSigma} upper-bounds the infection probabilities:

\begin{lemma}
For all $t\in[0, T]$, we have 
\begin{equation}\label{eq:pleqbarp}
p(t)\leq \bar p(t), 
\end{equation}
where $\bar p$ is the solution of the differential equation~\eqref{eq:barSigma}.
\end{lemma}

\begin{proof}
By Definition~\ref{defn:temporal}, there exist finitely many undirected graphs~$G_1$, \dots,~$G_L$ and real numbers $0 = t_0 < t_1< \cdots < t_L = T$ such that $\mathcal G(t) = G_\ell$ if $t_{\ell-1}\leq t<t_\ell$. Then, we can show (see, e.g.,~\cite{Ogura2015c}) that the differential equation ${dp}/{dt} \leq (BA_\ell - D) p$ holds true for $t_{\ell-1}\leq t<t_\ell$, where $A_\ell$ denotes the adjacency matrix of~$G_\ell$. Therefore, there exists an $\mathbb{R}^n$-valued function $\epsilon$ such that $\epsilon(t)\geq 0$ for all $t$ and 
\begin{equation*}
\frac{dp}{dt} = (BA_\ell - D) p - \epsilon,\quad t_{\ell-1}\leq t<t_\ell. 
\end{equation*}
Solving this differential equation for $0\leq t \leq t_1$ shows
\begin{equation*}
\begin{aligned}
p(t) 
&= e^{(BA_1 - D)(t)}p(0) - \int_0^t e^{(BA_1 - D)(t-\tau)} \epsilon(\tau)\,d\tau
\\
&\leq 
e^{(BA_1 - D)(t)}\bar p(0)
\\
&= \bar p(t), 
\end{aligned}
\end{equation*}
where we used the fact that $BA_1 - D$ is a Metzler matrix~\cite{Farina2000} and the initial condition~$\bar p(0) = p(0)$. Therefore, inequality~\eqref{eq:pleqbarp} holds true if $0\leq t\leq t_1$. Using an induction, we can extend the inequality for all $t \in [0, T]$.
\end{proof}

About the upper-bound $\bar p$ on the infection probabilities, we can prove the following proposition: 

\begin{proposition}\label{prop:prelim}
Let $t \in [0, T]$ and $i\in\{1, \dotsc, n\}$. For $\tilde \delta \in \mathbb{R}_{+}^n$, let us write $\hat \delta - \tilde \delta = (\hat \delta - \tilde \delta_1, \dotsc, \hat \delta - \tilde \delta_n) \in \mathbb{R}_+^n$. Then, the function
\begin{equation}\label{eq:element}
\mathbb{R}^{2n}_{+} \to \mathbb{R}_{+}\colon (\beta, \tilde \delta) \mapsto \bar p_i(t; \beta, \hat \delta -  \tilde \delta)
\end{equation}
is the pointwise limit of a sequence of posynomials over~$\mathbb{R}_{+}^{2n}$. 
\end{proposition}

\begin{proof}
Let $t\in [0, T]$ be arbitrary. Since the temporal network~$\mathcal G$ is  piecewise constant, there exist nonnegative numbers $h_1$, \dots,~$h_L$ such that $h_1 + \cdots + h_L = t$ and
\begin{equation*}
\begin{aligned}
\bar p(t; \beta, \delta) 
&= 
\exp((BA_L-D)h_L) \dotsm \exp((BA_1-D)h_1)  p(0)
\\
&=
\biggl( \prod_{\ell=1}^L \exp((BA_\ell-D)h_\ell) \biggr) p(0). 
\end{aligned}
\end{equation*}
Let $\tilde D$ be the diagonal matrix having the diagonals $\tilde \delta_1$, \dots,~$\tilde \delta_n$. Then, 
\begin{equation}\label{eq:barp(t)=...}
\begin{aligned}
\bar p(t; \beta, \hat \delta - \tilde \delta ) 
&= 
e^{-\hat \delta t}\biggl( \prod_{\ell=1}^L \exp((BA_\ell+\tilde D)h_\ell) \biggr) p(0)
\\
&=
\lim_{s\to\infty} f_{s}(\beta, \tilde \delta), 
\end{aligned}
\end{equation}
where 
\begin{equation*}
f_{s}(\beta, \tilde \delta) = e^{-\hat \delta t} \biggl( \prod_{\ell=1}^L
\sum_{k=0}^{s} \frac{{(BA_\ell+\tilde D)}^{k} h_\ell^{k}}{{k}!} \biggr)  p(0). 
\end{equation*}
Notice that all the entries of the matrix power~$(BA_\ell+\tilde D)^{s}$ are posynomials in the variable $(\beta, \tilde \delta) \in \mathbb{R}_{+}^{2n}$. Furthermore, the entries of the vector~$p(0)$ are positive. Therefore, any entry of the vectorial function $f_{s}$ is a posynomial. Hence, equation~\eqref{eq:barp(t)=...} shows that the mapping~\eqref{eq:element} is the pointwise limit of a sequence of posynomials, as desired.
\end{proof}

We are now ready to prove Theorem~\ref{thm:main}. 

\begin{proofof}{Theorem~\ref{thm:main}}
Assume that $b \in \mathbb{R}^n$ and $\tilde d \in \mathbb{R}^n$ solve the feasibility problem~\eqref{eq:convexConstraints}. Define $\beta$ and~$\delta$ by~\eqref{eq:solutionExpressions}. Then, by the definition of the function~$r^+$, we can show that the budget constraint~\eqref{eq:budgetConstraint} is satisfied. The constraints~\eqref{eq:convexConstraints3} and~\eqref{eq:convexConstraints4} immediately imply that the feasibility constraints~\eqref{eq:boxConstraints} are satisfied. Finally, by the definition of the functions $f$ and $F$, the first constraint~\eqref{eq:convexConstraints1} implies that
\begin{equation}\label{eq:pre:performanceConstraint}
J(\bar p) \leq \bar J. 
\end{equation}
On the other hand, by Assumption~\ref{asm:}, there exists a sequence $\{J_k\}_{k=0}^\infty$ of generalized finite-posynomials with nonnegative exponents such that \begin{equation}\label{eq:J=limk...}
J(p) = \lim_{k\to\infty}J_k(p)
\end{equation}
for all $p \in L^\infty([0, T], \mathbb{R}^n_{+})$. Since each $J_k$ has only nonnegative exponents, inequality~\eqref{eq:pleqbarp} shows $J_k(p) \leq J_k(\bar p)$. This inequality together with~\eqref{eq:J=limk...} and~\eqref{eq:pre:performanceConstraint} imply that the performance constraint~\eqref{eq:performanceConst} holds true. Therefore, the transmission and recovery rates given by~\eqref{eq:solutionExpressions} indeed solve Problem~\ref{prb:original}.

Let us show the convexity of the feasibility problem~\eqref{eq:convexConstraints}. It is sufficient to show that the functions $r^+$ and $f$ are convex. To show the convexity of~$r^+$, define the function $\tilde R^+(\beta, \tilde \delta) = \sum_{i=1}^n (\phi_i^+(\beta_i) + \tilde \psi_i^+(\tilde \delta_i))$, which is a posynomial by Assumption~\ref{asm:cost}. Since, for $b, \tilde d \in \mathbb{R}^n$, we have
\begin{equation*}
\begin{aligned}
r^+(b, \tilde d\tdsp )
&=
\log \sum_{i=1}^n (\phi_i^+(\exp(b_i)) + \psi_i^+(\hat \delta - \exp(\tilde d_i)))
\\
&=
\log \tilde R^+(\exp[b], \exp[\tilde d\tdsp ])
, 
\end{aligned}
\end{equation*}
Lemma~\ref{lem:llconvexity} shows that $r^+$ is convex.

Then, let us show the convexity of~$f$. We take a sequence~$\{J_k\}_{k=0}^\infty$ of generalized finite-posynomials with nonnegative exponents such that~\eqref{eq:J=limk...} holds true. Then, in the same way as in~\eqref{eq:def:F} and~\eqref{eq:def:f}, 
for each $k\geq 1$ we define the functions
\begin{align} \allowdisplaybreaks
F_k(\beta, \delta) &= J_k(\bar p(\cdot; \beta, \delta)), 
\label{eq:J(p)=lim...1}
\\
f_k(b, \tilde d\tdsp ) &= \log F_k(\exp[b], \hat \delta - \exp [\tilde d\tdsp ]). 
\label{eq:J(p)=lim...2}
\end{align}
Since $f$ is the pointwise limit of the sequence of functions~$\{f_k\}_{k\geq 0}$, it is sufficient to show the convexity of~$f_k$. Since $J_k$ is a generalized finite-posynomial with nonnegative exponents, there exist a positive integer $m_k$, indices~$i_{k1}, \ldots, i_{km_k} \in \{1, \ldots, n\}$, times~$t_{k1}, \ldots, t_{km_k}\in [0, T]$, and a generalized posynomial $g_k \colon \mathbb{R}^{m_k}_{+} \to \mathbb{R}_{+}$ with nonnegative exponents such that
\begin{equation}\label{eq:Jk(barp)=...} 
J_k(\bar p) = g_k(\bar p_{i_{k1}}(t_{k1}), \ldots, \bar p_{i_{km_k}}(t_{km_k})). 
\end{equation}
By Proposition~\ref{prop:prelim}, for each $j = 1, \dotsc, m_k$, there exists a sequence of posynomials $\{h_{kj}^{(\ell)} \}_{\ell=0}^\infty$ on $\mathbb{R}_{+}^{2n}$ such that
\begin{equation*}
\bar p_{i_{kj}}(t_{kj}; \beta,\hat \delta -  \tilde \delta) = \lim_{\ell\to\infty} h_{kj}^{(\ell)}(\beta, \tilde \delta). 
\end{equation*}
Therefore, by equation~\eqref{eq:Jk(barp)=...} and the continuity of~$g_k$, we obtain
\begin{equation}\label{eq:J_k(p)...}
J_k(\bar p(\cdot; \beta, \hat \delta -  \tilde \delta))  = \lim_{\ell\to\infty} \zeta_k^{(\ell)}(\beta, \tilde \delta), 
\end{equation}
where $\zeta_k^{(\ell)}(\beta, \tilde \delta) =  g_k(h_{k1}^{(\ell)}(\beta, \tilde \delta), \ldots, h_{km_k}^{(\ell)}(\beta, \tilde \delta))$. Notice that $\zeta_k^{(\ell)}$ is a generalized posynomial on $\mathbb{R}^{2n}_{+}$ because $g_k$ has nonnegative exponents and $h_{k1}^{(\ell)}$, $\dotsc$, $h_{km_k}^{(\ell)}$ are posynomials~\cite[Section~5.3]{Boyd2007}. Therefore, the mapping 
\begin{equation*}
\mathbb{R}^{2n} \to \mathbb{R} \colon (b, \tilde d\tdsp ) \mapsto \log \zeta_k^{(\ell)} (\exp[b], \exp[\tilde d\tdsp ])
\end{equation*}
is convex by Lemma~\ref{lem:llconvexity}. Since equations~\eqref{eq:J(p)=lim...1},  \eqref{eq:J(p)=lim...2}, and~\eqref{eq:J_k(p)...} show $$f_k(b, \tilde d\tdsp ) = \lim_{\ell\to\infty} \log \zeta_k^{(\ell)} (\exp[b], \exp[\tilde d\tdsp ]),$$ we obtain the convexity of~$f_k$, as desired.
\end{proofof}

\section{NUMERICAL SIMULATIONS}\label{sec:simulations}

In this section, we illustrate the obtained theoretical results by numerical simulations. We use the empirical temporal network of contacts between the children and teachers in a primary school~\cite{Stehle2011,Gemmetto2014}. In the school, each of the 5 grades is divided into two classes, for a total of 10 classes. Face-to-face interactions between children and teachers were recorded over two days. In this paper, we use the interaction data among the third-grade students on the first day. The resulting temporal network has $n = 44$ nodes and is defined from $t=0$ to $t = \num[group-separator={,}]{31110}$ [sec].

The cost functions for tuning the rates are set to be
\begin{equation*}
\phi_i(\beta_i) = c_{1i}+c_{2i}/{\beta_i^\lambda},
\ 
\psi_i(\delta_i) = c_{3i}+c_{4i}/{{(\hat \delta - \delta_i)}^\lambda}, 
\end{equation*}
where $\hat \delta$ is a constant greater than $\max(\bar \delta_1, \dotsc, \bar \delta_n)$, $\lambda$ is a positive parameter for tuning the shape of the cost functions, and $c_{1i}$, \dots,~$c_{4i}$ are constants to normalize the cost functions as $\phi_i(\ubar \beta_i) = 1$, $\phi_i(\bar \beta_i) = 0$, $\psi_i(\ubar \delta_i) = 0$, and $\psi_i(\bar \delta_i) = 1$. Under this normalization, we have $R(\beta, \delta) = 0$ if $(\beta_i, \delta_i) = (\bar \beta_i, \ubar \delta_i)$ for every node~$v_i$ (i.e., all nodes keep their ``nominal'' infection and transmission rates), while $R(\beta, \delta)=2n$ if $(\beta_i, \delta_i) = (\ubar \beta_i, \bar \delta_i)$ for every $i$ (i.e., all nodes receive the full amount of vaccinations and antidotes).

In this simulation, we let $\ubar \beta_i = 5\times 10^{-4}$, $\bar \beta_i = 5\times 10^{-3}$, $\ubar \delta_i = 10^{-4}$, $\bar \delta_i = 10^{-3}$, $\hat \delta = 10$, and $\lambda=10^{-2}$. We assume that $p(1)  = \cdots = p(11) = 1$ and $p(12) = \cdots = p(44) = 1/100$, i.e., the nodes $v_1$, \dots, $v_{11}$ are infected, while other nodes $v_{12}$, \dots, $v_{44}$ are highly susceptible at time $t=0$. In order to protect the initially susceptible nodes, we use the performance measure $J(p)= \sum_{i=12}^n p_i(T)$. The performance measure obviously satisfies Assumption~\ref{asm:}. Under the budget constraint~$R(\beta, \delta) \leq n,$ we find sub-optimal solutions to Problem~\ref{prb:bc} (the budget-constrained allocation problem) by solving the convex optimization problem~\eqref{eq:opt:bc}. For comparison, we find the sub-optimal transmission and recovery rates that minimize the decay rate of the infection probabilities of the SIS model~\cite{Preciado2014} over a time-aggregated static network using the same cost functions and the budget constraint. As the time-aggregated network, we use the weighted and undirected graph where the weight of an edge is equal to the frequency of the edge appearing in the temporal network.

\begin{figure}[tb]
\vspace{2mm}
\centering
\includegraphics[width=1\linewidth]{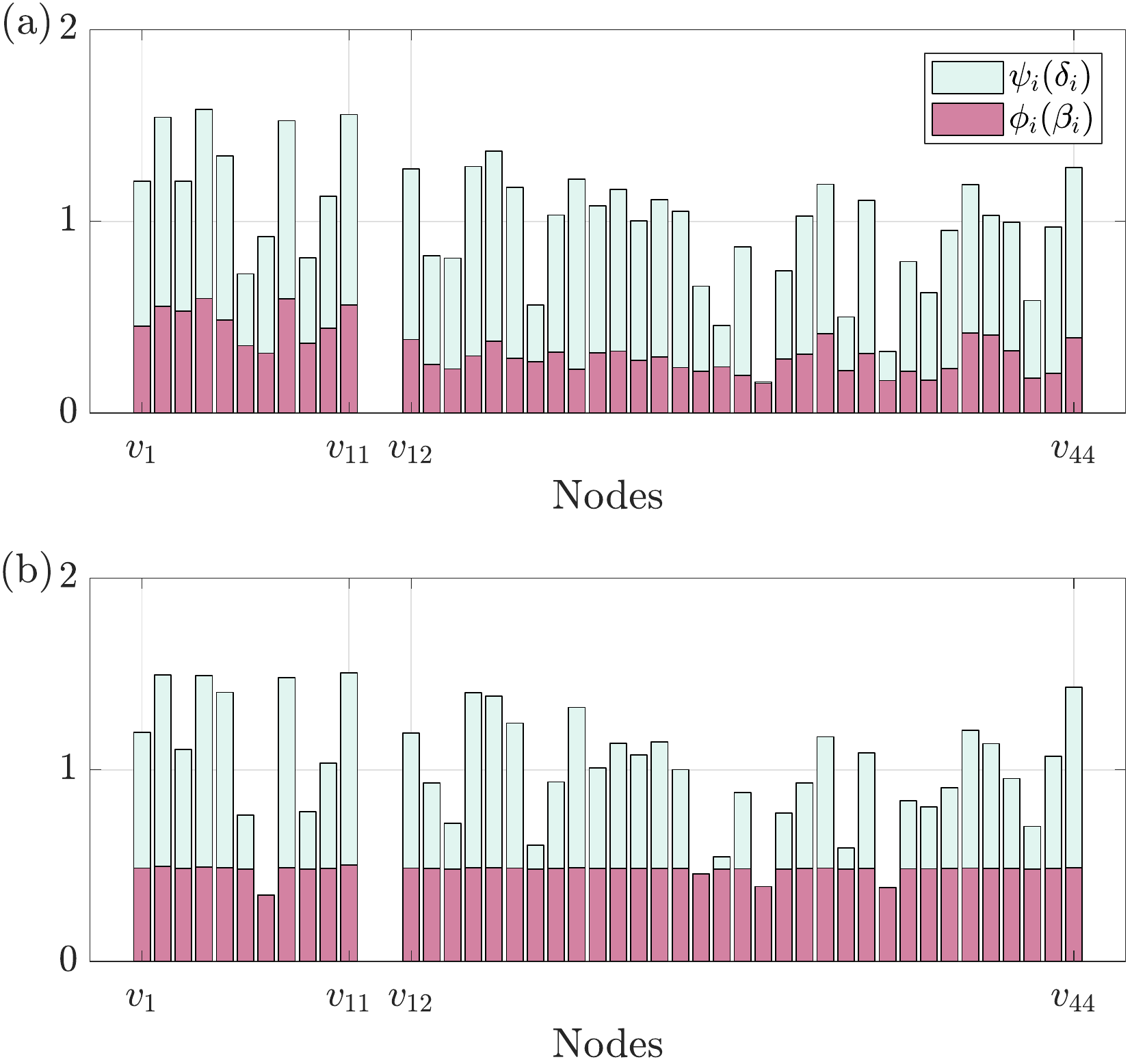}
\caption{Investments on nodes when $v_1$, \dots, $v_{11}$ are initially infected. (a) Sub-optimal solutions to the budget-constrained allocation problem ($J(p) \leq 1.17$). (b) Sub-optimal solutions to the budget-constrained allocation problem \cite[Problem~3]{Preciado2014} for the time-aggregated static network ($J(p) \leq 19.5$).}
\label{fig:optimalDistribution}
\end{figure}

Using Theorem~\ref{thm:main}, we verify the performances of the investments. The proposed investments guarantee $J(p) \leq 1.17$ and drastically improve the one~$J(p) \leq 19.5$ from the method for the time-aggregated static network. In Fig.~\ref{fig:optimalDistribution}, we compare the investments from the proposed and the conventional methods. We see that the proposed method invests in reducing the transmission rates in a heterogeneous manner, while the conventional investments on the transmission rates are almost equal among nodes.

\section{CONCLUSIONS} 

In this paper, we have presented a computationally efficient framework for determining the distribution of control resources for eradicating epidemic outbreaks in empirical temporal networks. We have shown that the resource distribution problem can be reduced to a convex feasibility problem. In the reduction, the posynomials with nonnegative exponents have played an important role. We have illustrated the obtained theoretical results with numerical simulations on the temporal network of contacts within a primary school.

\end{document}